\documentclass[%
 article,two column,
superscriptaddress,
preprintnumbers,
 amsmath,amssymb,
 aps,
]{revtex4-2}
\usepackage{graphicx}
\usepackage{dcolumn}
\usepackage{bm}
\usepackage{bbold}
\usepackage{calc}
\usepackage{algpseudocode}
\usepackage{changes}
\usepackage[braket, qm]{qcircuit}
\usepackage{xcolor}
\usepackage{braket}
\usepackage{subfig}
\usepackage{amsthm}
\usepackage{ulem}
\usepackage{hyperref}
\usepackage{url}
\newtheorem{theorem}{Theorem}
\newtheorem{corollary}{Corollary}[theorem]
\newtheorem{lemma}{Lemma}[theorem]

\usepackage{color}
\usepackage{xcolor,colortbl}
\usepackage{comment}
\newcommand {\nc} {\newcommand}
\nc {\IR} [1]{\textcolor{red}{#1}}
\nc {\IB} [1]{\textcolor{blue}{#1}}
\nc {\IM} [1]{\textcolor{magenta}{#1}}

\hypersetup{
    colorlinks=true,
    linkcolor=blue,
    citecolor=blue,
    filecolor=magenta,      
    urlcolor=blue,
    pdftitle={Overleaf Example},
   pdfpagemode=FullScreen,
    }

\begin{document}

\preprint{N3AS-24-010}

\preprint{RIKEN-iTHEMS-Report-24}

\preprint{LA-UR-24-22665}

\title{Exact block encoding of imaginary time evolution with universal quantum neural networks}
\author{Ermal Rrapaj}
\email{ermalrrapaj@lbl.gov}
\affiliation{Lawrence Berkeley National Laboratory, One Cyclotron RD, Berkeley, CA, 94720, USA}
\affiliation{Department of Physics, University of California, Berkeley, CA 94720, USA}
\affiliation{RIKEN iTHEMS, Wako, Saitama 351-0198, Japan}

\author{Evan Rule}
\email{erule@berkeley.edu}
\affiliation{Department of Physics, University of California, Berkeley, CA 94720, USA}
\affiliation{Theoretical Division, Los Alamos National Laboratory, Los Alamos, NM 87545, USA}




\begin{abstract}
We develop a constructive approach to generate quantum neural networks capable of representing the exact thermal states of all many-body qubit Hamiltonians. The Trotter expansion of the imaginary-time propagator is implemented through an exact block encoding by means of a unitary, restricted Boltzmann machine architecture. Marginalization over the hidden-layer neurons (auxiliary qubits) creates the non-unitary action on the visible layer. Then, we introduce a unitary deep Boltzmann machine architecture, in which the hidden-layer qubits are allowed to couple laterally to other hidden qubits. We prove that this wave function ansatz is closed under the action of the imaginary-time propagator and, more generally, can represent the action of a universal set of quantum gate operations.
We provide analytic expressions for the coefficients for both architectures, thus enabling exact network representations of thermal states without stochastic optimization of the network parameters. In the limit of large imaginary time, the ansatz yields the ground state of the system. The number of qubits grows linearly with the system size and total imaginary time for a fixed interaction order. Both networks can be readily implemented on quantum hardware via mid-circuit measurements of auxiliary qubits. If only one auxiliary qubit is measured and reset, the circuit depth scales linearly with imaginary time and system size, while the width is constant. Alternatively, one can employ a number of auxiliary qubits linearly proportional to the system size, and circuit depth grows linearly with imaginary time only.
\end{abstract}

\maketitle

\section{Introduction}
Quantum many-body systems are associated with exponentially large Hilbert spaces, which create fundamental challenges for even the most powerful supercomputers. In many cases, the goal is to understand the properties of the ground state or first few excited states, the dynamical evolution of the system given some initial conditions, and/or the thermal behavior at some finite temperature. To tackle such problems, one usually defines a variational ansatz for the wave function or density matrix, with parameters tuned to represent (as close as possible) the physical system under study. Some examples in this regard range from mean-field approximations to more complicated methods such as tensor network ~\cite{White:2003,Verstraete:2004,Vidal:2007,Verstraete:2008,SCHOLLWOCK:2011,Orus:2014} and neural network states~\cite{LeCun:2015,Carleo:2017,Deng:2017,Carrasquilla:2017,Hibat:2020}. For the ansatz to be efficient, the number of free parameters should increase at most as a polynomial function of the number of particles (or degrees of freedom) in the system. In addition, one also needs to devise learning methods to optimize these variational parameters. Many optimizers have been tested in practice, including gradient-based methods like conjugate gradient~\cite{Hestenes:1952}, adaptive moment estimation~\cite{Kingma:2017}, simultaneous perturbation stochastic approximation ~\cite{Spall:1992}, and the natural gradient~\cite{Amari:1998}, as well as gradient-free methods such as Nelder-Mead~\cite{Nelder:1965}. 

A class of very powerful schemes for studying many-body problems is represented by Monte Carlo (MC) methods, which allow for the calculation of properties of quantum states through sampling. These methods are particularly useful when the wave function can be represented by a positive probability density up to an overall normalization constant~\cite{Yokoyama:1987,Giamarchi:1991,KOONIN:1997,Sorella:2005,Roggero:2014,Wlazowski:2014,Carlson:2015,Tews:2020}. The partition function of the interacting problem is cast into a sum (or integral) over configurations in a chosen basis and, with a suitable importance-sampling scheme, one can (in principle) cover the exponentially large configuration space in polynomial time. 

Given a Hamiltonian $H$, one method for determining the ground state is to evolve the system with the imaginary-time propagator
\begin{equation}
    \ket{\Psi(\tau)}=e^{-\tau H}\ket{\Psi_0}
    \label{eq:im_time_prop}
\end{equation}
for some large time $\tau>0$. Provided the initial state $\Psi_0$ is not orthogonal to the ground state and an energy gap $\Delta E$ exists between the ground state(s) and the first excited state(s), then the evolved state $\Psi(\tau)$ converges to the ground state of $H$ for $\tau \gg 1/\Delta E$. Thus, the problem of determining the ground state is equivalent to finding an efficient representation of the action of the imaginary-time propagator on the initial state.

The Diffusion Monte Carlo (DMC) and Green’s Function Monte Carlo
(GFMC) methods are ways of exactly solving the ground state by means of the stochastic procedure of MC sampling~\cite{Kalos:1962,Anderson:1975,Ceperley:1980}. The ground state of many physical systems of interest is characterized by correlations among particles, and therefore the initial-state ansatz $|\Psi_0 \rangle$ should also contain correlations. As the particle number increases, so do the degrees of freedom, and the number of parameters required to specify the ansatz grows very large. For instance, in the case of nuclear Hamiltonians the interaction contains quadratic spin and isospin and tensorial operators, and the many body wave function cannot be written
as a product of single-particle spin-isospin states. By considering all possible nucleon pairs in the systems, the number of possible
spin-states grows exponentially with the number of nucleons. To perform a DMC calculation with standard nuclear Hamiltonians, it
is then necessary to sum over all possible single-particle spin-isospin states of the
system to build the trial wave function used for propagation. This is the standard
approach in GFMC calculations for nuclear systems~\cite{Gandolfi:1987,Pudliner:1995,Piper:2005,Carlson:2003,Carlson:2015}.

An alternative approach is the Auxiliary Field Monte Carlo (AFMC) method, where the imaginary-time propagator $e^{-\tau H}$ is rewritten to reduce the dependence on spin and isospin operators from quadratic to linear by using the Hubbard-
Stratonovich transformation~\cite{SCHMIDT:1999,Rubenstein:2012,LeBlanc:2015,Lee:2021,Lonardoni:2018,Shi:2021}. By introducing a Gaussian auxiliary field, the quadratic operators can be expressed as integrals over one-body operators between the physical fields and the auxiliary ones. The linear form of the operators allows one to write the trial wave function $|\Psi_0\rangle$ as a product of
single-particle states. The sampling of auxiliary fields required to perform the integral produces the same end result as the propagator with quadratic spin-isospin operators acting on a trial wave function containing all the possible good spin-isospin states. Introducing the auxiliary field reduces the number of terms in the trial wave function from exponential to linear, with the additional computational cost of sampling of the integral over auxiliary fields.

However, the infamous sign problem --- wherein the weights of configurations in the MC simulation become negative or even complex, and therefore cannot be interpreted as classical probabilities --- prevents these methods from efficiently solving a wide class of problems. In such cases, the MC simulation requires an exponential number of samples to obtain a satisfactory signal-to-stochastic-noise ratio~\cite{Loh:1990,HATANO:1992,Troyer:2005,Kaul:2013}. 

In recent years, neural networks have become a powerful tool to represent complex correlations in multivariable functions and probability distributions and have found wide applications in academia and industry through the popularity of deep learning methods. Neural-network quantum states (NQS) have attracted significant interest as powerful wave-function ansatzes. Restricted Boltzmann machine (RBM), an energy-based, shallow generative neural network capable of representing arbitrary distribution functions~\cite{Roux:2008}, has been used for dimensionality reduction~\cite{Hinton:2006}, classification~\cite{Larochelle:2008}, feature learning~\cite{Coates:2011}, and, most importantly to the present study, to represent quantum many-body wave functions and propagators~\cite{Carleo:2017,Melko:2019,Rrapaj:2020txq}.
Explicitly, the RBM architecture represents the wave function in the computational basis in terms of $N$ visible-layer qubits $\vec{z}$ and $M$ hidden-layer qubits $\vec{h}$,
\begin{equation}
\begin{split}
    \langle \vec{z} | \Psi \rangle &=\Psi_\mathcal{R}(\vec{z})\\
    &=\sum_{\vec{h}}\exp\left[\sum_{i}a_iz_i+\sum_{i,j}z_iW_{ij}h_j+\sum_{j}b_jh_j\right],
    \label{eq:RBM_ansatz}
\end{split}
\end{equation}
as specified by a set of coupling parameters that we denote collectively by $\mathcal{R}\equiv\left(\vec{a},\overset{\text{\tiny$\bm\leftrightarrow$}}{W},\vec{b}\right)$. In general, the coupling coefficients are allowed to be complex numbers. The summation
\begin{equation}
    \sum_{\vec{h}}\equiv \sum_{h_1=\pm 1}\ldots\sum_{h_M=\pm 1}
\end{equation}
denotes the marginalization over all hidden qubits. An RBM encoding of a physical state is said to be efficient if the number of required hidden units scales polynomially with the size of the physical system.

A distinct feature of RBMs is that they can be analyzed by employing tensor network theory~\cite{Clark:2018,Glasser:2018}, which has been used to show that the entanglement entropy of the ansatz scales with a subregion’s volume~\cite{Chen:2018} rather than its area, a behavior typical of tensor network states. Indeed, RBMs can efficiently and exactly represent many highly entangled states~\cite{Gao:2017}, including chiral states (such as the Laughlin state~\cite{Glasser:2018}), arbitrary graph states~\cite{Raussendorf:2001}, topological toric code states~\cite{KITAEV:2003}, and states obeying the entanglement volume law~\cite{Verstraete:2006}.

On the other hand, there exists a class of states, which are either ground states of gapped Hamiltonians or can be generated by a polynomial-size quantum circuit, that the RBM architecture cannot efficiently represent \cite{Gao:2017}. If the wave-function ansatz of Eq. (\ref{eq:RBM_ansatz}) is expanded to include one more hidden layer --- resulting in a Deep Boltzmann machine (DBM) --- then such states can be efficiently represented \cite{Gao:2017}, provided that the coupling coefficients are complex-valued. For the transverse-field Ising model and the Heisenberg model, one can construct DBM networks with a polynomially-scaling number of qubits that can exactly represent the imaginary time evolution through Trotter expansion~\cite{Carleo:2018}. The resulting DBM parameters are complex, further complicating the use of classical MC sampling, as discussed above.

Due to the recent rapid developments in quantum technology, these wave-function ansatzes have been ported to quantum circuits, with the Variational Quantum Eigensolver (VQE) being one of the most prominent examples~\cite{Peruzzo:2014,McClean:2016,TILLY:20221}, paving the way for variational hybrid quantum-classical algorithms. Additionally, the non-unitary imaginary time evolution of an initial state can be approximated by unitary operators in a quantum circuit. There are many methods that fall into this category, and are commonly referred to as Quantum Imaginary Time Evolution (QITE)~\cite{McArdle:2019,Motta:2020,Turro:2021,Huang:2023}. Some methods optimize the variational ansatz to match the effect of the imaginary time evolution, while other methods employ a block-encoding scheme that performs the desired non-unitary operation in a probabilistic fashion. In the former case, the quality of the ansatz determines the quality of the approximation, and in the latter, post-selection is needed to remove instances when the desired operation has not been performed. 



In this work, we provide several important contributions to performing QITE for all possible spin Hamiltonians. This is achieved by means of auxiliary qubits on quantum hardware, in analogy with AFMC methods. In Sec. \ref{sec:RBM}, we represent the Trotter step of any imaginary-time propagator of Pauli strings through a unitary RBM. In contrast to the wave-function ansatz in Eq. (\ref{eq:RBM_ansatz}), the RBM propagator that we develop has full Pauli support, including $\sigma^x, \sigma^y$. 

In Sec. \ref{sec:DBM}, we construct a universal, unitary DBM ansatz capable of exactly representing any thermal state within the Trotter approximation scheme. The DBM wave function that we consider is diagonal in the computational basis. We provide the DBM parameters as functions of the couplings in the original Hamiltonian, removing the need for variational optimization. In other words, we devise the DBM representation of an arbitrary initial state and compute the resulting parameters from the action of the propagator in terms of Trotter steps. The imaginary time evolution represents the thermal ensemble at finite temperature, and for very long times it purifies out the ground state. Both Boltzmann machines can be efficiently implemented on quantum hardware, and we provide analysis of the respective circuit widths and gate counts. In Sec. \ref{sec:numerics}, we provide a numerical example of our method, simulating the resulting quantum circuit and computing observables by means of sampling. We conclude by summarizing our main results in Sec. \ref{sec:conclusion}. We expand upon and explain further details in a series of appendices at the end of this article. In this work, we refer to $z$ to denote the eigenvalue of $\sigma^z$, and $\sigma$ to generically represent Pauli operators.


\section{Auxiliary field decomposition of imaginary-time propagator}
\label{sec:RBM}


A priori, the Hamiltonian $H$ need not be expressed in terms of interactions between spin-$1/2$ qubits. For example, starting with a system of $N_F$ interacting fermions occupying $N_S$ available single-particle states, we can use one of several well-known encoding schemes (e.g., Jordan-Wigner \cite{1928ZPhy...47..631J}, Bravi-Kitaev \cite{2002AnPhy.298..210B}) to express $H$ in terms of a system of $N$ qubits. Depending on the nature of the fermionic system and the choice of encoding, the resulting representation of $H$ in the qubit basis may contain $k$-local interactions for $k\leq N$. Hereafter, we shall assume that $H$ is expressed in terms of a system of $N$ qubits. 

For such a system, the imaginary-time propagator can be represented in terms of an RBM with \textit{real} coefficients \cite{Rrapaj:2020txq}, providing an efficient classical algorithm for sampling the ground state (or ground-state expectation values of few-body operators). In the present work, our aim is modify this approach to be suitable for quantum computation by constructing a representation of the imaginary-time propagator in terms of auxiliary qubits with \textit{imaginary} coupling coefficients. The resulting RBM is composed of unitary operators, which are straightforward to implement on quantum hardware. 

Our algorithm relies upon several theorems:
\\

\begin{theorem}
For any $\tau \in \mathbb{R}$, the exponential of the diagonal $M$-qubit Pauli string $e^{- \tau \prod_{i=1}^M \sigma^z_i}$ can be expressed as a marginalization over two-qubit gates between the visible qubits $\sigma_i^z$ and a single auxiliary qubit, plus induced Pauli strings of length $k<M$.
\label{thm:pauli_strings}
\end{theorem}
\begin{proof}
Let $C(M,k)$ denote the set of $k$-element combinations drawn from the set of indices $I=\{1,...,M\}$. The most general Hamiltonian formed from Pauli $\sigma^z$ strings of length $k\leq M$ can be expressed as
\begin{equation}
\tau H=\sum_{k=0}^M\sum_{P\in C(M,k)}K_P^{(k)}\prod_{j \in P}\sigma^z_j,
\label{eq:h_lhs}
\end{equation}
where $K_P^{(k)}$ are the $2^M$ possible coupling parameters, and we include the $k=0$ term corresponding to an overall energy shift. The interaction is diagonal in the computational basis as we only consider $\sigma^z_i$. 
 
 
 Our ansatz for the marginalized form is
\begin{equation}
e^{-\tau H}=\sum_{h=\pm 1}e^{-ih\left(C+\sum_{r = 1}^MW_r\sigma^z_r\right)},
\label{eq:rbm_ansatz}
\end{equation}
where the coefficients $W_r$ control the coupling between the hidden qubit $h$ and the visible qubits, and the constant $C$ represents a bias term for $h$. One expects an overall norm to be present on the right hand side, but this factor can be absorbed into the $k=0$ term in $H$, so we do not explicitly show it here.

Taking the logarithm of both sides, in the computational basis, yields the matching condition
\begin{widetext}
\begin{equation}
-\sum_{k=0}^M\sum_{P\in C(M,k)}K_P^{(k)}\prod_{j \in P}z_j=\ln\left\{\sum_{h=\pm 1}\exp\left[-ih\left(C+\sum_{r=1}^MW_rz_r\right)\right]\right\},
\label{eq:matching}
\end{equation}
where we have replaced the operator $\sigma_r^z$ by its eigenvalue $z_r$.
\end{widetext}
 It is straightforward to interpret this constraint as a system of linear equations, expressed in matrix form as
\begin{equation}
\mathcal{M}\vec{K}=\vec{L}.
\end{equation}
The matrix $\mathcal{M}$ corresponds to the basis of Pauli $\sigma^z$ strings of length $k\leq M$ and is therefore guaranteed to be nonsingular, and $\vec{L}$ is the right-hand side of Eq. \eqref{eq:matching}. Thus, for fixed values of the parameters $C$, $W_r$ it is always possible to find couplings $K_P^{(k)}$ that satisfy Eq. (\ref{eq:matching}). Our aim, however, is to find auxiliary coupling parameters $C$, $W_r$ that reproduce a desired value for the highest-order coupling $K^{(M)}$. (We drop the subscript $P$ as there is a unique highest-order coupling). Upon inverting Eq. (\ref{eq:matching}), we find that $K^{(M)}$ is determined by the logarithm of a ratio of products of cosine functions with different frequencies. Within the domain of the logarithm, by tuning the parameters $C$ and $W_r$, one may fix the value of the highest order coupling $K^{(M)}$ to any desired real value. A detailed derivation can be found in Appendix~\ref{app:rbm_coupling}. 
\end{proof}

The resulting lower-order couplings $K_P^{(k)}$ for $k <M$ can first be transferred to the right-hand side of Eq. (\ref{eq:rbm_ansatz}) (with a negative sign) and then recursively expressed in terms of additional auxiliary fields. Thus, it will require $\lesssim 2^M$ auxiliary fields to express a generic $M$-local operator in the RBM architecture, assuming all coefficients are non-zero. This recursive process is represented schematically in Fig. \ref{fig:rbm_identities} for the cases $M=1$, $2$, and $3$. As discussed in detail below, when $M\leq 2$ it is possible to choose the RBM couplings so that no lower-order terms are induced. 

In this work, we frequently rotate between the $x$, $y$, and $z$ bases using the operators
\begin{equation}
           H^x=\frac{1}{\sqrt{2}}\left(\begin{array}{cc}
            1 & 1 \\
             1 & -1
           \end{array}\right),~~~H^y=\frac{1}{\sqrt{2}}\left(\begin{array}{cc}
            - i & i \\
              1 & 1
           \end{array}\right).
       \end{equation}
$H^x$ is the usual Hadamard operator, and $H^y$ is an analogous single-qubit operator defined to satisfy the crucial identities
\begin{equation}
    \sigma_i^x=H^x_i  \sigma_i^z H^x_i,~~~~\sigma_i^y=H^y_i \sigma_i^z {H^y_i}^{\dag},
    \label{eq:hadamard_rotations}
\end{equation}
where ${H^x_i}^{\dagger}=H^x_i$.
With the above expressions, it is now straightforward to generalize Theorem \ref{thm:pauli_strings}:
\begin{corollary}
For any $\tau\in \mathbb{R}$, the exponential of an arbitrary $M$-qubit Pauli string $e^{-\tau \prod_{i=1}^M \sigma_i}$ can be expressed as a marginalization over two-body gates between the visible qubits and a single auxiliary qubit, plus induced Pauli strings of length $k<M$.
\label{cor1}
\end{corollary}
\begin{proof}
We can independently rotate each qubit using Eq.~\eqref{eq:hadamard_rotations} so that the respective Pauli string is $\sigma^z$. Then, we apply  the result of Theorem \ref{thm:pauli_strings} [i.e., Eq.~\eqref{eq:rbm_ansatz}], and finally rotate back to the original Pauli string. This produces the same identity as in Eq.~(\ref{eq:rbm_ansatz}), but the Pauli strings on both sides of the equation are now generic.
\end{proof}

In general, there are infinitely many choices of $C$ and $W_r$ that yield the same value of $K^{(M)}$. It is possible to choose the hidden-layer couplings to simplify or even eliminate some of the induced couplings. For example, in the case of a two-local interaction it is possible to marginalize over an auxiliary qubit $h$ without inducing any one-body interactions. Specifically, we can write
\begin{equation}
\begin{split}
e^{-K^{(2)} \sigma_1 \sigma_2} =& A \sum_{h=\pm1} e^{-iW(\sigma_1+s \sigma_2)h},\\
=& 2 A\ \left( \mathbb{1} \otimes   \langle \pm 1|\right) e^{-iW(\sigma_1+s \sigma_2)\sigma^x_h} \left(\mathbb{1} \otimes | \pm 1\rangle\right),
\end{split}
\label{eq:2body_id}
\end{equation}
where the normalization $A = e^{K^{(0)}}$, and $\mathbb{1}$ is the identity operator in the physical Hilbert space. The normalization $A$ and the coupling constants $W$, $s$ are given by
\begin{equation}
\begin{split}
A&=\exp(| K^{(2)}|)/2,\\
 W&=\frac{1}{2} \cos^{-1}\left[\exp(-2|K^{(2)}|)\right],\\ 
 s &= \mathrm{sign}(K^{(2)}).
 \end{split}
 \label{eq:2bd}
\end{equation}
In accordance with Corollary \ref{cor1}, the above identity holds for arbitrary Pauli operators $\sigma_i=\{\mathbb{1}_i,\sigma^x_i,\sigma^y_i,\sigma^z_i\}$. When one of the Pauli strings is the single-qubit identity $\mathbb{1}_i$, the expression reduces to an RBM auxiliary-field identity for non-unitary one-qubit operators 
\begin{equation}
    e^{-K^{(1)}\sigma_1}=A\sum_{h=\pm 1}e^{-iW(\sigma_1 + s\mathbb{1})h},
    \label{eq:1bd}
\end{equation}
where $A$, $W$, and $s$ are given by Eq. \eqref{eq:2bd} with the one-body coupling $K^{(1)}$ in place of $K^{(2)}$.

\begin{figure*}
    \centering
    \includegraphics[scale=0.5]{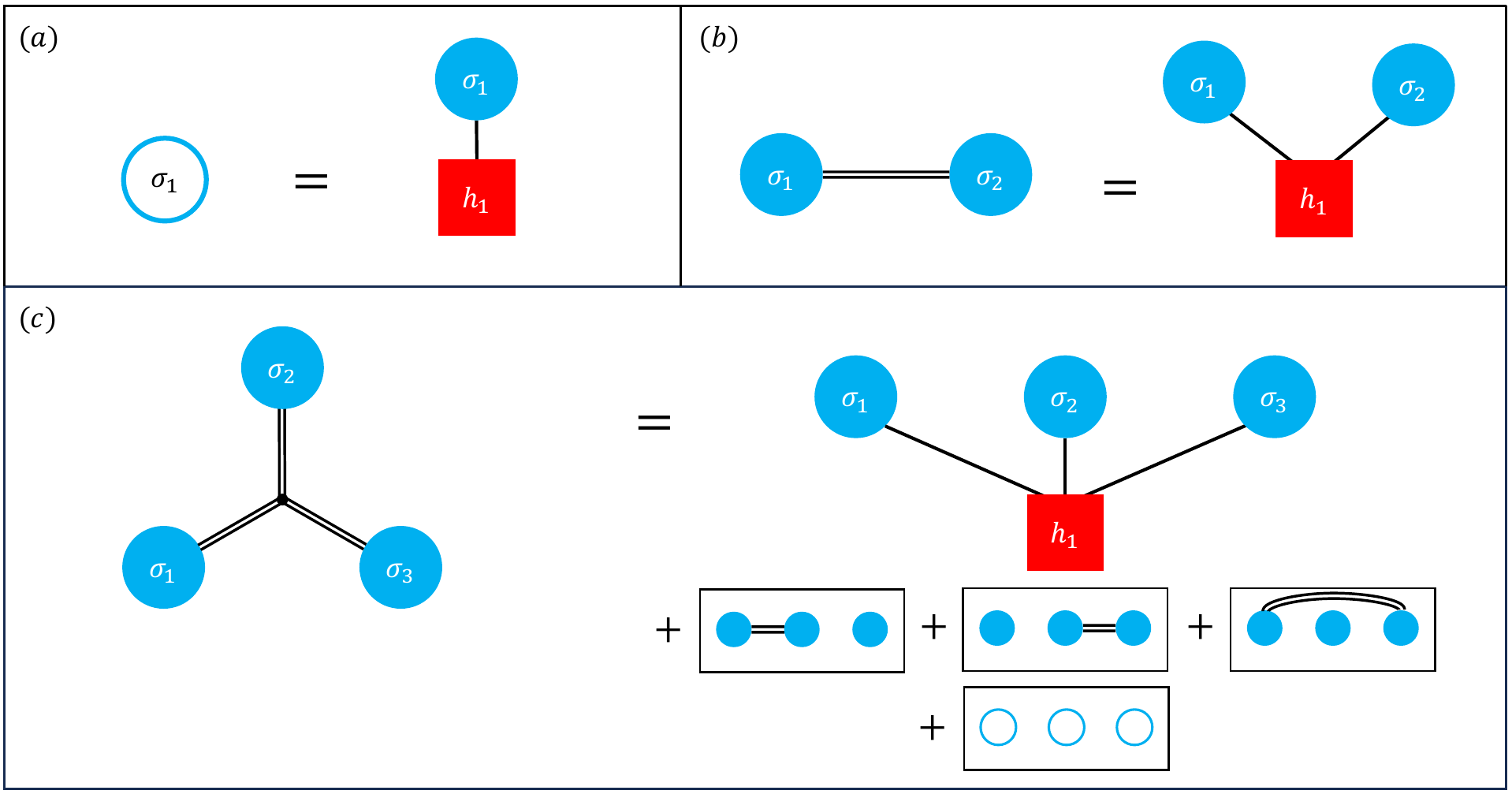}
    \caption{Graphical representations of the RBM identities for (a) one-body, (b) two-body, and (c) three-body interactions. Open circles denote non-unitary one-body interactions. Double lines represent non-unitary interactions between two or more qubits. Single lines denote unitary interactions between visible and hidden qubits. In sub-figure (c), the inset boxes show the induced one- and two-body couplings between the three visible qubits. Within each inset, the ordering of the RBM indices is the same as above, although the labels have been omitted for simplicity.}
    \label{fig:rbm_identities}
\end{figure*}

The right-hand side of Eq. (\ref{eq:2body_id}) can be performed by explicit summation over the auxiliary qubit $h$ (see the first line), resulting in a linear combination of unitaries (LCU) \cite{Childs2012HamiltonianSU}. As the second line shows, the unitary operator can also be expressed as an exact block encoding~\cite{Gilyen:2019} of the imaginary-time propagator. This provides an alternative route by initializing the ancilla qubit in $\ket{\pm1}$, and post-selecting on $\ket{\pm1}$ (see also~\cite{hsieh2021unitary}). Both the LCU and block-encoding methods are non-deterministic, and in case of failure, one must repeat the entire circuit. 

The probability of successfully encoding the auxiliary qubit depends on the initial state of the visible qubits
\begin{equation}
P_s=1-(1-e^{- 4 |K^{(2)}|})\alpha,
\label{eq:prob_success}
\end{equation}
where $\alpha\equiv P(\sigma_1=s \sigma_2)$ is the probability of measuring $\sigma_1=s\sigma_2$ in the initial state. For small values of $|K^{(2)}|$, the probability of success is close to 1, and for very large values $P_s \approx 1-\alpha$. Eq. (\ref{eq:prob_success}) holds for the one-body case with $\alpha = P(\sigma_1=s)$.
By marginalizing over all possible physical states, it is easy to see that the average value is 
\begin{equation}
\langle P_s \rangle = \frac{1}{2}(1+e^{-4 |K^{(2)}|}). 
\label{eq:avg_Ps_2bd}
\end{equation}

The two-body interaction implemented by Eq. (\ref{eq:2body_id}) is a single component of a larger total Hamiltonian. If we restrict our focus to the two-qubit subsystem with ``local'' Hamiltonian $H=K^{(2)}\sigma_1\sigma_2$, then the two (degenerate) ground states correspond to the anti-aligned case $\sigma_1 = -s \sigma_2$, and the two (degenerate) excited states correspond to the aligned case $\sigma_1=s\sigma_2$. 

If the system is initialized in the anti-aligned state, then $\alpha=0$ and Eq. (\ref{eq:prob_success}) implies that the auxiliary encoding cannot fail. On the other hand, if the system is initially in the aligned state, then $\alpha = 1$ and the probability of success is $e^{-4|K^{(2)}|}$. This represents the worst-case scenario, as $P_s\rightarrow 0$ for large values of $K^{(2)}$. We can view our procedure as a type of acceptance-rejection algorithm: if the two-qubit subsystem is in the ground state of the ``local'' Hamiltonian, then the result is accepted with 100\% probability. If the subsystem is in an excited state, then the sample is accepted with a probability that decreases exponentially with the size of the local energy gap. 

With every application of the ancilla method, the overall probability of success, being the product of the individual probabilities, will decrease exponentially, which is a typical feature of block-encoding methods. When Eq. (\ref{eq:2body_id}) is employed in a Trotter approximation scheme, choosing a small time step to control the Trotter error naturally leads to a small value of $|K^{(2)}|$, which in turn increases the probability of successfully implementing the auxiliary qubit identity.


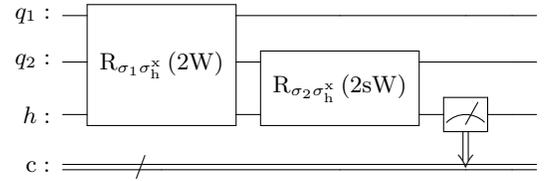
\begin{figure}
    \centering
\begin{equation*}
\Qcircuit @C=1.0em @R=1.0em { \\
	 	\nghost{{q}_{1} :  } & \lstick{{q}_1 :  } & \multigate{2}{\mathrm{R_{\sigma_1 \sigma^x_h}}\,(\mathrm{2 W})} & \qw & \qw & \qw & \qw\\
	 	\nghost{{q}_{2} :  } & \lstick{{q}_{2} :  } & \ghost{\mathrm{R_{\sigma_1 \sigma^x_h}}\,(\mathrm{2 W})} & \multigate{1}{\mathrm{R_{\sigma_2 \sigma^x_h}}\,(\mathrm{2 s W})} & \qw & \qw & \qw\\
	 	\nghost{h :  } & \lstick{h :  } & \ghost{\mathrm{R_{\sigma_2 \sigma^x_h}}\,(\mathrm{2 W})} & \ghost{\mathrm{R_{\sigma_2 \sigma^x_h}}\,(\mathrm{2s W})} & \meter & \qw & \qw\\
	 	\nghost{\mathrm{{c} :  }} & \lstick{\mathrm{{c} :  }} & \lstick{/_{_{}}} \cw & \cw & \dstick{_{_{\hspace{0.0em}}}} \cw \ar @{<=} [-1,0] & \cw & \cw\\
\\ }
\end{equation*}
\caption{Quantum circuit diagram that implements the two-body identity in Eq. \eqref{eq:2body_id} using block encoding.}
\label{fig:circ_2bd}
\end{figure}
The circuit in Fig. \ref{fig:circ_2bd} implements Eq. (\ref{eq:2body_id}) using the block-encoding method. The two-qubit rotation gates are followed by a measurement of the auxiliary qubit $h$ for post-selection, and the value is stored in the classical bit $c$.



In the case of a three-body interaction, the marginalization ansatz is
\begin{equation}
\begin{split}
\exp\bigg[-&\Big(K^{(3)} \sigma_1 \sigma_2 \sigma_3 + K^{(2)}_{12} \sigma_1 \sigma_2 + K^{(2)}_{13} \sigma_1 \sigma_3 \\
 +& K^{(2)}_{23} \sigma_2 \sigma_3
+ K^{(1)}_1 \sigma_1 + K^{(1)}_2 \sigma_2 + K^{(1)}_3 \sigma_3 \Big) \bigg]\\
=& A \sum_{h=\pm1} e^{-i( W_1 \sigma_1+ W_2 \sigma_2 + W_3 \sigma_3)h - i W_0 h}\\
= 2 A ( \mathbb{1} \otimes   &\langle \pm 1| ) e^{-i( W_1 \sigma_1+ W_2 \sigma_2 + W_3 \sigma_3)\sigma^x_h - i W_0 \sigma^x_h} (\mathbb{1} \otimes | \pm 1\rangle ).
\label{eq:3bdy_rbm}
\end{split}
\end{equation}
One possible solution is
\begin{equation}
\begin{split}
W\equiv W_1=W_2=W_3&=\frac{1}{2}\tan^{-1}\left[\left(1-e^{-8| K^{(3)}|}\right)^{1/4}\right],\\
W_0&=s W,\\
A&=\frac{1}{2}\left[\sec^4(2W)\sec(4W)\right]^{1/8},
\end{split}
\end{equation}
where $s=\mathrm{sign}(K^{(3)})$. We find that it is necessary in this case to include a bias term $W_0$ for the hidden unit. These choices reproduce the three-local term, but they also induce one- and two-local contributions:
\begin{equation}
    \begin{split}
    K^{(1)}_{1}=K^{(1)}_{2}=K^{(1)}_{3}&=-\frac{s}{8}\log[\cos(4W)],\\
        K^{(2)}_{12}=K^{(2)}_{13}=K^{(2)}_{23}&=-\frac{1}{8}\log[\cos(4W)].
    \end{split}
\end{equation}
In practice, the induced one- and two-local terms could be absorbed into the corresponding intrinsic couplings and then expressed through auxiliary fields using the one- and two-local identities of Eqs. \eqref{eq:1bd} and \eqref{eq:2body_id}, respectively. This iterative procedure, starting from the highest-order interaction, absorbing the induced couplings into the physical couplings at lower orders, and then repeating marginalization, produces $e^{-\tau H}$ entirely expressed in terms of a sum of unitary operators. 

Like the two-body case, the encoding of the non-unitary three-body operator is non-deterministic, and the probability of success depends on the initial state of the visible qubits. In particular,
\begin{equation}
\begin{split}
    P_s&=1-\sin^2(2W)\alpha_2-\sin^2(4W)\alpha_4\\
    &=1-\left[1-f_2\left(|K^{(3)}|\right)\right]\alpha_2+\left[1-f_4\left(|K^{(3)}|\right)\right]\alpha_4,
    \end{split}
\end{equation}
where
\begin{equation}
\begin{split}
    f_2(x)&\equiv \frac{1}{1+\sqrt{1-e^{-8x}}},\\
    f_4(x)&\equiv \frac{e^{-16x}}{\left(1+\sqrt{1-e^{-8x}}\right)^4},
    \end{split}
\end{equation}
and $\alpha_n\equiv P(|\sigma_1+\sigma_2+\sigma_3+s|=n)$. As before, if the system is initially in one of the degenerate ground states of the ``local'' Hamiltonian 
\begin{equation}
H=K^{(3)}\sigma_1\sigma_2\sigma_3+K_{12}^{(2)}\sigma_1\sigma_2+\ldots,
\end{equation} 
then $\alpha_2=\alpha_4=0$, and the block encoding succeeds with 100\% probability. If the system is prepared in one of the degenerate first excited states, then $\alpha_2=1$, $\alpha_4=0$, and the result is accepted with a probability that exponentially approaches $1/2$ as the local energy gap increases. Finally, in the worst-case scenario, the system is initially in the unique second excited state; then $\alpha_2=0$, $\alpha_4=1$, and the result is accepted with a probability that exponentially approaches zero as $K^{(3)}\rightarrow \infty$. 

Averaging over all possible physical states, the typical success probability is 
\begin{equation}
\begin{split}
    \braket{P_s}&=\left[3+4\cos^2(2W)+\cos^2(4W)\right]/8\\
    &=\left[3+4f_2\left(|K^{(3)}|\right)+f_4\left(|K^{(3)}|\right)\right]/8,
    \end{split}
\end{equation}
which approaches $5/8$ for large values of $K^{(3)}$. 

The RBM identity for a non-unitary four-body operator is presented in Appendix \ref{app:4bdy_id}. We briefly summarize this result by noting that it is possible to choose the RBM couplings so that only two-body couplings are induced, without introducing any one- or three-body terms. Even so, it can quickly become very expensive to implement higher-body interaction terms, especially if the induced couplings cannot be absorbed into existing intrinsic interactions.

If new auxiliary qubits are used for every time step, the width of the circuit increases linearly with the number of Trotter steps and with system size (number of interactions), but the number of auxiliary qubits can be exponential in the order of the interaction,
resulting in a significant cost in terms of the number of qubits required.

An alternative approach is to express the non-unitary many-body interaction in terms of a product of $H^x$, $H^y$, and $CX$ gates, plus a single one-body non-unitary operator for which we may then implement the RBM identity in Eq. \eqref{eq:1bd}. This method does not induce lower-order operators, at the cost of additional $CX$'s and increased circuit depth. 

Then, just one auxiliary qubit suffices, by measuring and resetting the qubit after each propagator application (two-qubit rotation). This implementation requires sequential application of all two-qubit gates, resulting in a deep circuit, the depth of which scales linearly with system size and number of Trotter steps. An intermediate scenario would utilize all available ancillary qubits to parallelize as many two-qubit gates as possible, resetting the qubits before the next sequence of propagators. For a local Hamiltonian, the number of auxiliary qubits would be linearly proportional to system size, and the depth of the circuit would scale linearly with the number of Trotter steps.

Such a route is feasible due to the following theorem:
\begin{theorem}
For any $\tau \in \mathbb{R}$, the exponential of the $M$-qubit Pauli string $e^{- \tau \prod_{i=1}^M \sigma_i}$ can be expressed as $U e^{-\tau \sigma^z_M} U^{\dagger}$, where the unitary matrix $U$ is a product of $H^x$, $H^y$, and $CX$ gates.
\label{thm:pauli_strings_$CX$}
\end{theorem}
\begin{proof}
The derivation for the unitary operator $e^{- i \tau \prod_{i=1}^M \sigma_i}$ with $ \tau \in \mathbb{R}$ can be found in Ref. \cite{Whitfield:2011}. Here we provide a quick summary and develop the non-unitary equivalent: One can independently rotate each qubit with $H^x$ or $H^y$ gates so that the Pauli string is composed entirely of $\sigma^z$ operators. Intuitively, since the resulting Pauli string $\prod_i\sigma_i^z$ has two eigenvalues $\pm 1$, it can be replaced by a single Pauli operator, nested within a series of $CX$ gates. In particular
\begin{equation}
    e^{- \tau \prod_{i=1}^M \sigma^z_i}=\prod_{i=1}^{M-1}CX_{(i,i+1)}e^{-\tau \sigma^z_M}\prod_{i=M-1}^1CX_{(i,i+1)},
\end{equation}
where $CX_{(i,j)}$ is a control X gate with control on qubit $i$ and target qubit $j$. 
\end{proof}
As an explicit example, the three-body $XXZ$ operator can be expressed as 
\begin{equation}
\begin{split}
    e^{-K^{(3)}\sigma_1^x\sigma_2^x\sigma_3^z}=A~U \left[\sum_h e^{-iW(\sigma_3^z+s\mathbb{1})h}\right]U^{\dag},
\end{split}
\end{equation}
where 
\begin{equation}
\begin{split}
 U&=H^x_1~H^x_2~CX_{12}~CX_{23},\\
 A&=\exp(|K^{(3)}|)/2,\\
 W&=\frac{1}{2}\cos^{-1}\left[\exp(-2|K^{(3)}|)\right],\\
 s&=\mathrm{sign}(K^{(3)}).
 \end{split}
\end{equation}
Thus, an $N$-body non-diagonal, non-unitary operator can be expressed in terms of unitary operators using only a single auxiliary qubit.
\section{Boltzmann machine wave function ansatz}
\label{sec:DBM}



So far we have demonstrated that the imaginary-time propagator $e^{-\tau H}$ can be represented by an RBM with imaginary coefficients, ideal for quantum gate implementations. The resulting operator would then be acted on an initial wave function $\Psi_0$ to produce the thermal state of $H$, as in Eq. (\ref{eq:im_time_prop}). Here, we introduce the unitary DBM ansatz 
consisting of $N$ visible-layer qubits, $M$ hidden-layer qubits, and $M'$ deep-layer qubits 
\begin{equation}
    \begin{split}
    \langle \vec{z} | \Psi \rangle=& \Psi_{\mathcal{D}}(\vec{z})\\
    =&\sum_{\vec{h},\vec{d}}\exp \bigg[ i \bigg(\sum_i a_i z_i + \sum_{i,j}z_i W_{ij}h_j \\
    &~~~~~~~+ \sum_{i,j}h_i W'_{ij}d_j + \sum_{i}b_i h_i + \sum_{i}b'_i d_i \bigg) \bigg].
    \end{split}
    \label{eq:DBM_ansatz}
\end{equation}
The shorthand $\mathcal{D}=(\vec{a},\vec{b},\vec{b}',\overset{\text{\tiny$\bm\leftrightarrow$}}{W},\overset{\text{\tiny$\bm\leftrightarrow$}}{W'})$ denotes the complete set of DBM couplings, which are limited to visible/hidden and hidden/deep interactions as well as one-body terms for each layer. Crucially, there are no lateral connections. An overall factor of the imaginary unit $i$ has been set aside in the exponentials, so that when all coupling parameters are real the DBM is a sum of manifestly unitary operators.

\begin{figure*}
    \centering
    \includegraphics[scale=0.55]{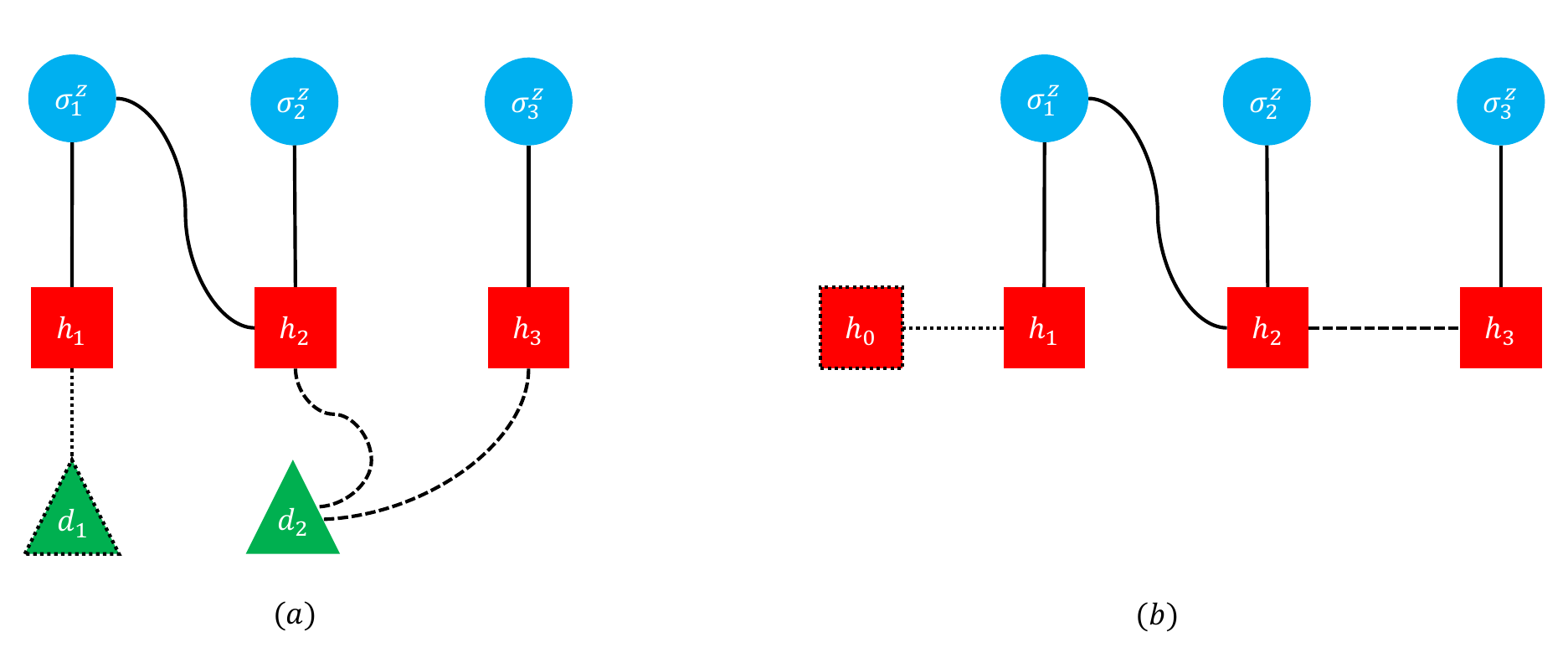}
    \caption{Two representations of the same physical system of 3 visible qubits $\sigma_i^z$ expressed in (a) DBM ansatz of Eq. (\ref{eq:DBM_ansatz}) and (b) L-DBM ansatz of Eq. (\ref{eq:DBM_lat_ansatz}). Lines between visible (hidden) and hidden (deep) units represent interlayer couplings $W_{ij}$ ($W'_{ij}$). Lines between hidden units correspond to lateral couplings $L_{ij}$. Dashed lines denote encodings that are equivalent after marginalization over hidden unit $d_2$. Dotted lines indicate that deep unit $d_1$ in the DBM is relabeled as laterally coupled hidden unit $h_0$ in the L-DBM architecture. Additional one-body bias terms ($a_i,b_i,b_i'$) are not represented explicitly.}
    \label{fig:dbm_ldbm_diagrams}
\end{figure*}

In this work, we introduce a unitary, semi-restricted Boltzmann machine
\begin{equation}
    \begin{split}
    \Psi_{\mathcal{L}}(\vec{z})=&\sum_{\vec{h}}\exp \bigg[ i \bigg(\sum_i a_i z_i + \sum_{i,j}z_i W_{ij}h_j \\
    &~~~~~~~~~+ \sum_{i<j}h_i L_{ij}h_j + \sum_{i}b_i h_i \bigg) \bigg],\\
    \end{split}
    \label{eq:DBM_lat_ansatz}
\end{equation}
consisting of $N$ visible qubits and $M$ hidden qubits with lateral couplings $L_{ij}$ between hidden-layer qubits. Note that there are still no direct interactions between visible qubits. The shorthand $\mathcal{L}=(\vec{a},\vec{b},\overset{\text{\tiny$\bm\leftrightarrow$}}{W},\overset{\text{\tiny$\bm\leftrightarrow$}}{L})$ represents the complete set of couplings. We refer to the wave function ansatz of Eq. (\ref{eq:DBM_lat_ansatz}) as a lateral-DBM (L-DBM). Although no auxiliary qubits are explicitly labeled as ``deep'', we will later demonstrate that the L-DBM is equivalent to a particular class of DBMs after performing marginalization. Graphical depictions of equivalent DBM and L-DBM networks are shown in Fig. \ref{fig:dbm_ldbm_diagrams} (a) and (b), respectively. 

\begin{theorem}
The L-DBM wave function ansatz of Eq. (\ref{eq:DBM_lat_ansatz}) is universal in the sense that it can represent any wave function resulting from imaginary time propagation of an arbitrary initial L-DBM state $\Psi_{\mathcal{L}_0}(\vec{z})$.
\label{thm:dbm_universal}
\end{theorem}
\begin{proof}
    If the interaction $H$ is diagonal in the computational basis, then one may use the recursive procedure described in the previous section to write the imaginary-time propagator $e^{-\tau H}$ in terms of the RBM ansatz, which is itself a special case of the L-DBM.

       For a generic Hamiltonian $H$, we will rotate the imaginary-time propagator to the computational basis using the operators $H^x_i$ and $H^y_i$, as in Eq. (\ref{eq:hadamard_rotations}). Thus, it suffices to show that the action of each rotation operator $(H^x_i,H^y_i,{H^y_i}^{\dag})$ on the L-DBM wave function can be expressed as a new L-DBM, possibly with additional auxiliary qubits and modified couplings.

\begin{lemma}
  \label{lemma:Hx_DBM}
The L-DBM ansatz is closed under the action of the operator $H^x_l$ on visible qubit $l$
 \begin{equation}
     H^x_l\Psi_{\mathcal{L}_0}(\vec{z})=A\Psi_{\mathcal{L}_1}(\vec{z}),
     \label{eq:Hx_action}
 \end{equation}
 where $A$ is a normalization constant.
 \end{lemma}

The proof is provided in Appendix~\ref{app:Hadamard_action}.
The action of $H^x_l$ on the L-DBM wave function can be summarized as follows:
\begin{enumerate}
\item Introduce one new hidden qubit $h_{M+1}$ with couplings
\begin{equation}
\begin{split}
W_{l,M+1}&=\pi/4\\
L_{j,M+1} &= - W_{lj} \\
b_{M+1}&=-\left(a_l+\frac{\pi}{4}\right)\\
\label{eq:L-DBM_updated_couplings}
   \end{split}
\end{equation}
\item Sever all existing couplings between visible qubit $l$ and the hidden layer: $W_{lj}\rightarrow 0$
\item Update the one-body term at index $l$: $a_l\rightarrow \pi/4$
\item The normalization factor is $ A=\exp(-i\pi/4)/\sqrt{2}$.
\end{enumerate}
This procedure is illustrated for a simple L-DBM architecture in Fig. \ref{fig:rbm_hadamard}. We note in passing the similarity of our approach to that outlined for standard DBM networks in Ref.~\cite{Carleo:2018}. However, there are significant distinctions: unlike the reference, we devise procedures that work generically (for all operators), and our method preserves the unitarity of the architecture.
In other words, if the initial L-DBM is unitary (that is, if the couplings are all real), then the updated L-DBM is also unitary. 
(The normalization factor $A$ is generically complex, but it need not be implemented on quantum hardware.)

An analogous statement to Lemma \ref{lemma:Hx_DBM} can be proven for the rotation operators $H^y$ and ${H^y}^\dag$:
\begin{lemma}
\label{lemma:Hy_DBM}
The L-DBM ansatz is closed under the action of the rotation operators $H^y_l$ and ${H^y_l}^{\dag}$ on visible qubit $l$
 \begin{equation}
 \begin{split}
     H^y_l\Psi_{\mathcal{L}_0}(\vec{z})&=A_1\Psi_{\mathcal{L}_1}(\vec{z}),\\
     {H^y_l}^\dag\Psi_{\mathcal{L}_0}(\vec{z})&=A_2\Psi_{\mathcal{L}_2}(\vec{z}),
     \end{split}
     \label{eq:Hy_action}
 \end{equation}
 for some normalization constants $A_1$, $A_2$ and real-valued L-DBM parameters $\mathcal{L}_1$, $\mathcal{L}_2$.
 \end{lemma}
This statement can be proven using the same technique employed for $H^x_l$. The algorithmic steps required for updating the L-DBM in these cases are given in Appendix \ref{app:hadamard}. Each rotation operator requires the introduction of a single auxiliary qubit. The number of auxiliary qubits (hidden units) is equal to the number of interactions. Thus, for a local Hamiltonian the scaling of the L-DBM is linear with system size and also linear in the number of Trotter steps.

\begin{figure*}
    \centering
    \includegraphics[scale=0.65]{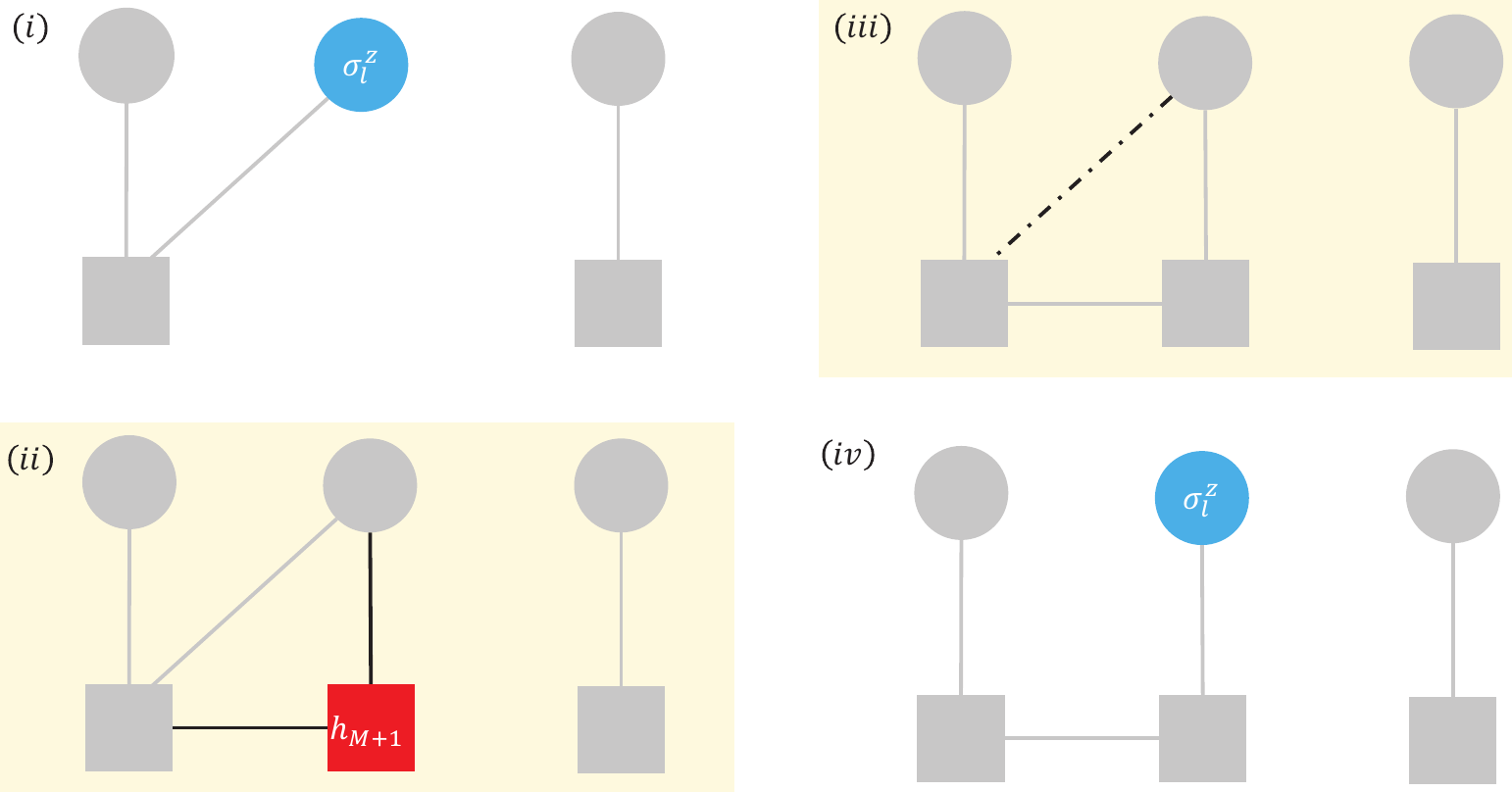}
    \caption{Illustration of the action of the rotation operators $(H, H^y, {H^y}^{\dag})$ on the L-DBM ansatz. (i) Rotation operator is applied to visible qubit $l$. (ii) Add one hidden qubit $h_{M+1}$ that couples to visible qubit $l$ and couples laterally to all hidden units currently connected to visible qubit $l$. (iii) Remove all couplings between visible qubit $l$ and existing hidden layer qubits. (iv) Adjust the one-body coupling of visible qubit $l$. The exact values of the required couplings differ depending on the operator under consideration (see text).}
    \label{fig:rbm_hadamard}
\end{figure*}

By rotating to the computational basis using the rotation operators and then applying Theorem \ref{thm:pauli_strings}, we can effectively write
\begin{equation}
    e^{-\tau H}\Psi_{\mathcal{L}_0}(\vec{z})=A\Psi_{\mathcal{L}_1}(\vec{z}).
\end{equation}
Thus, the L-DBM ansatz can represent any state obtained by imaginary time evolution of an initial L-DBM state.
\end{proof}
In fact, we can prove that an even stronger statement about the generality of the L-DBM ansatz:

\begin{theorem}
\label{thm:ldbm_universal_gates}
    The L-DBM ansatz (with real coefficients) is a universal representation of quantum states.
    \begin{proof}
        Single- and double-qubit rotations furnish a universal set of quantum operators. Clearly, a rotation of qubit $i$ by angle $\phi$ with respect to the $z$ axis can be absorbed in the L-DBM ansatz by taking $a_i\rightarrow a_i + \phi$. By Lemmas \ref{lemma:Hx_DBM} and \ref{lemma:Hy_DBM}, the L-DBM ansatz is also closed under single-qubit rotations around the $x$ and $y$ axes. 
        
        In the same way, it suffices to prove that the 2-qubit rotation $\exp(-i\phi\sigma_1^z\sigma_2^z)$ can be absorbed into the L-DBM wave function. One particular representation is
        \begin{equation}
           \begin{split}
           &e^{-i\phi\sigma_1^z\sigma_2^z} \\
           &= A e^{i\frac{\pi}{4}(\sigma_1^z+\sigma_2^z)}\sum_{h_1,h_2}\exp\bigg[i\frac{\pi}{4}\left(\sigma_1^z+\sigma_2^z\right)h_1+i\frac{\pi}{4}h_1h_2\\
           &\hspace*{4cm}-i\frac{\pi}{4}h_1+i\left(\phi+\frac{\pi}{4}\right)h_2\bigg],
        \end{split}
        \end{equation}
        where
        \begin{equation}
        \begin{split}
            A&=\frac{1}{2}e^{i\pi/4}.\\
            \end{split}
        \end{equation}
We can summarize the action of $e^{-i\phi \sigma_1^z\sigma_2^z}$ in terms of the L-DBM ansatz as
\begin{enumerate}
    \item Shift the existing values of the visible qubit bias terms $a_{i}\rightarrow a_{i} +\pi/4$ for $i=1,2$.
    \item Introduce 2 new hidden qubits $h_1$, $h_2$ with real coupling parameters
    \begin{equation}
    \begin{split}
        W_{11}&=W_{21}=L_{12}=-b_1=\pi/4,\\
        b_2&=\phi+\pi/4.
        \end{split}
    \end{equation}
\end{enumerate}
    \end{proof}
\end{theorem}
Combining the results of Theorems \ref{thm:dbm_universal} and \ref{thm:ldbm_universal_gates}, one can use arbitrary quantum gate operations to prepare a desired initial state $\ket{\Psi_0}$ and then represent the action of imaginary time evolution upon that state, all within the L-DBM ansatz.

\subsection{L-DBM as a DBM}
In the previous section, we introduced the L-DBM architecture and showed that it is universal. In fact, the L-DBM is a special case of the usual DBM ansatz. 

\begin{theorem}
    The L-DBM wave-function ansatz [Eq. (\ref{eq:DBM_lat_ansatz})] is a special case of the DBM ansatz [Eq. (\ref{eq:DBM_ansatz})].
\end{theorem}
\begin{proof}
    The L-DBM ansatz contains hidden units that couple laterally to other hidden units as well as to the visible layer. Let $h_j$ be a particular hidden unit of the L-DBM. We can distinguish 3 possible cases:
    \begin{enumerate}
    \item  If $L_{ij}=0$ for all $i\neq j$, then hidden unit $j$ has no lateral couplings. We identify it as a standard hidden unit in a DBM.
    \item If $W_{ij}=0$ for all $1\leq i\leq N$, then hidden unit $j$ does not couple to the visible layer. We identify it as a standard deep unit in a DBM.
    \item If hidden unit $j$ couples to the visible layer and to other hidden units, we can remove each visible layer coupling $z_iW_{ij}h_j$ by introducing a new hidden unit $\tilde{h}$. It follows from Eq. (\ref{eq:2body_id}) that we can write
    \begin{equation}
e^{z_iW_{ij}h_j}=A\sum_{\tilde{h}}e^{\tilde{W}\left(z_i+h_j\right)\tilde{h}}
    \end{equation}
    for some coefficients $A$, $\tilde{W}$. Recursively applying this argument, we can completely decouple $h_j$ from the visible layer. We can then reinterpret $h_j$ as a deep unit of the DBM, which couples only to the hidden layer.
    \end{enumerate}

\end{proof}

\section{Numerical example: Transverse Ising model}
\label{sec:numerics}
As an illustration, in this section we implement the unitary RBM representation of the imaginary-time propagator for the one-dimensional transverse Ising model, 
\begin{equation}
    \begin{split}
    H =&\sum_{i} \sigma^z_i \sigma^z_{i+1} - \sum_i \sigma^x_i
    \end{split}
\end{equation}
for a system of three qubits with periodic boundary conditions. Starting from the initial state $\ket{\Psi_0}=\ket{+1_x}\ket{+1_x}\ket{+1_x}$, we perform a series of second-order Trotter steps
\begin{equation}
   U^{(2)}(d\tau)=e^{d\tau/2 \sum_i \sigma^x_i}e^{-d\tau \sum_i \sigma^z_i \sigma^z_{i+1}}e^{d\tau/2 \sum_i \sigma^x_i}. 
\end{equation}
with a single auxiliary qubit that is reset after every application of the RBM identity.
We set the imaginary time step to be $d\tau=0.01$.

Using the Qiskit python package~\cite{Qiskit}, we simulate the resulting quantum circuit for various values of the total imaginary time $\tau$. For each value of $\tau$, we perform $10^6$ simulations in 100 batches of $10^4$ runs each. Within each batch, we compute the expectation values for the energy $\braket{H}$ and the two terms in the Hamiltonian $\braket{ZZ}$ and $\braket{X}$. Then, across all 100 batches, we compute the mean of each expectation value and estimate the corresponding uncertainty of the mean using the jackknife method~\cite{Miller:1974}. We verified that using the bootstrap method results in essentially the same confidence interval, and that combining the samples from all batches into a single batch results in the same confidence interval. The results are shown in Fig.~\ref{fig:Ising_numerics} (a), where the displayed uncertainty corresponds to one standard deviation about the mean value. Despite the large number of simulations in each batch, the actual number of valid samples used for statistics is given by the product of the number of simulations multiplied by the acceptance rate. For instance, for $\tau=1$ after post-selection, the number of effective samples is $\approx 57$ per batch, giving rise to relatively high uncertainty.  On near-term quantum devices, we expect higher uncertainty due to hardware noise, and we leave such studies to future work. 
\begin{figure}[h]
\subfloat[Expectation value of the Hamiltonian as a function of imaginary time $\tau$. The black points indicate the mean value across all 100 batches of simulations performed at each value of $\tau$, and the error bars show the corresponding uncertainty (one standard deviation). The black dashed line is the exact expectation value of the imaginary-time-evolved state.  Insets show the respective values for the $X$ and $ZZ$ terms in the Hamiltonian. The dotted line denotes the ground-state energy. ]{\includegraphics[width=0.5\textwidth]{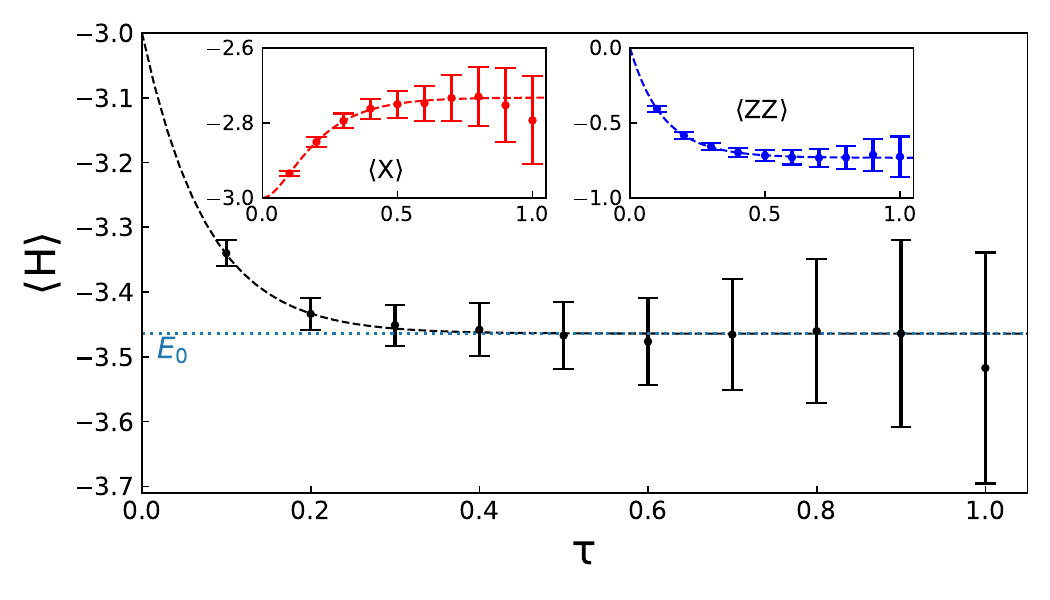}
}\\
\subfloat[Cumulative success probability of the auxiliary-qubit encoding as a function of imaginary time $\tau$. Black circles denote the empirical value determined from our numerical simulations, where it is computed by the ratio of the number of samples where the identity was corrected implemented to the total number of samples. For comparison, the dotted line shows the average success probability given by Eq. \eqref{eq:avg_Ps_ising}. The empirical success probability can be well approximated by the dashed line ($P_s(\tau)=10^{-\lambda \tau}, \ \lambda = \ln(10)$).]{\includegraphics[width=0.5\textwidth]{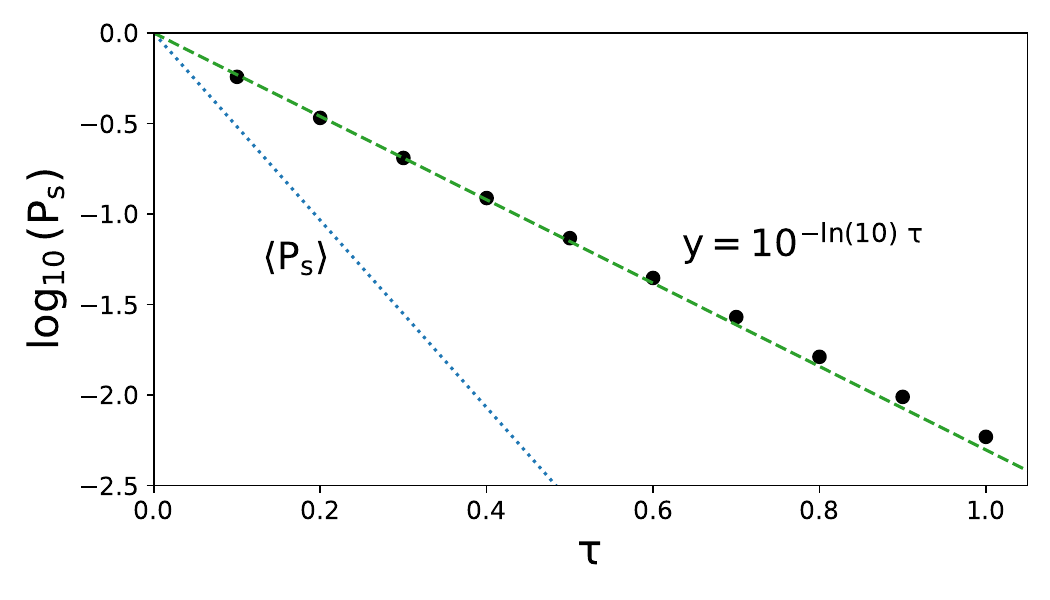}}
\caption{One dimensional transverse Ising model at the critical point with three qubits and periodic boundary conditions at finite temperature. (a) Energy expectation value as function of imaginary time (inverse temperature) and (b) success probability for post-selection.} 
\label{fig:Ising_numerics}
\end{figure}

As the probability of success decreases with the number of imaginary time steps, we display in Fig. \ref{fig:Ising_numerics} (b) the fraction of samples accepted with post-selection. As expected, the acceptance rate decreases exponentially with imaginary time. For comparison, we also compute the average success probability

\begin{equation}
\langle P_s\rangle(\tau) =[\langle P_s^x \rangle^6 (d\tau/2) \langle P_s^z\rangle^3 (d\tau)]^{\tau/d\tau},
\label{eq:avg_Ps_ising}
\end{equation}
where $\braket{P_s^x}(d\tau/2)$ and $\braket{P_s^z}(d\tau)$ are given by Eq. \eqref{eq:avg_Ps_2bd} with $K^{(2)}=d\tau/2$ and $d\tau$, respectively. As Fig. \ref{fig:Ising_numerics} (b) shows, the actual success probability is much better than the average success rate, which is due to the choice of the initial wave function. Indeed, $\Psi_0$ is the ground state of the $\sigma^x$ terms in the Hamiltonian. This means that the first application of the auxiliary qubit identity will never fail, and subsequent applications will succeed with much higher than average probability. The empirical values can be approximated by $P_s(\tau)=10^{-\lambda \tau}, \ \lambda = \ln(10)$, depicted by the dashed line.

\section{Conclusion and Discussion}
\label{sec:conclusion}
In recent years, improvements in quantum devices have paved the way for quantum computing technology and the development of algorithms to tackle problems otherwise unfeasible on classical computers. One such case is the study of quantum many-body systems characterized by a Hilbert-space size that scales exponentially with the degrees of freedom. Monte Carlo methods have proven successful for probing the thermal and/or ground-state properties of such systems, but they are limited due the sign problem. Performing similar computations on quantum hardware should be free of such disadvantages. Quantum computing methods like VQE employ a variational ansatz with free parameters that require optimization, in similar fashion to neural networks in classical machine learning. 

In this work, we devised a universal ansatz for the imaginary-time propagator and wave function of all qubit systems by means of unitary RBM and DBM architectures, respectively.  By explicit construction, we proved their universality and provided expressions for the parameters that require no optimization. Our ansatzes can be analyzed through the lenses of block encoding and auxiliary qubits, and we provided post-selection success probabilities as functions of the initial state and non-unitary coupling parameters. 
The number of auxiliary qubits grows linearly with the system size and total imaginary time. Both networks can be easily implemented on quantum hardware by measuring and resetting auxiliary qubits. As such, only one auxiliary qubit is actually needed, and the resulting circuit depth scales linearly with imaginary time and system size, while the width is constant. Alternatively, one can implement a number of auxiliary qubits linearly proportional to system size, and circuit depth grows linearly with imaginary time only.

\begin{acknowledgments}
We thank Kenneth McElvain, Alexander Kemper, and Omar Alsheikh for fruitful discussions and suggestions. ER is supported by the National Science Foundation under cooperative agreement 2020275 and by the U.S. Department of Energy through the Los Alamos National Laboratory. This research utilized resources of the National Energy Research Scientific Computing Center (NERSC), a Department of Energy Office of Science User Facility. 
\end{acknowledgments}

\appendix

\section{Expression for $K^{(M)}$ in terms of RBM couplings}
\label{app:rbm_coupling}
Suppose $M$ is even and let $C=0$, $W_r=W$ for all $1\leq r\leq M$. Upon inverting Eq. (\ref{eq:matching}), the highest-order coupling parameter $K^{(M)}$ can be expressed as
\begin{equation}
\begin{split}
    K^{(M)}&=\frac{1}{2^M}\sum_{k=0}^M (-1)^{k}~^MC_k\log\big\{2\cos\big[(2k-M)W\big]\big\}\\
    &=\frac{1}{2^M}\log\Bigg\{\frac{\prod_{k~\mathrm{even}}\cos[(2k-M)W]^{^MC_k}}{\prod_{k~\mathrm{odd}}\cos[(2k-M)W]^{^MC_k}}\Bigg\},
    \end{split}
    \label{eq:KM_solve}
\end{equation}
where $^MC_k$ is the usual binomial coefficient. The right-hand side is zero when $W=0$. The highest-frequency cosine function will appear in the numerator, the term with $k=M$; the first zero of this function occurs at $W=\pi/(2M)$, at which point $K^{(M)}$ diverges to $-\infty$. The right-hand side of Eq. (\ref{eq:KM_solve}) is continuous for $0\leq W < \pi/(2M)$, and all other cosine factors are strictly positive in this domain. Thus we can always find a value of $W$ that reproduces any $K^{(M)}$ in the range $(-\infty,0]$. 

For values of $K^{(M)}>0$, we flip the sign of a single RBM coupling $W_s \rightarrow -W$ relative to the other $M-1$ couplings. The net effect of this change is to invert the ratio in the second line of Eq. (\ref{eq:KM_solve}), changing the range of $K^{(M)}$ to $[0,+\infty)$.

When $M$ is odd, if we choose $C=W_r=W$ for all $1\leq r\leq M$, we recover the same form as Eq. (\ref{eq:KM_solve}) with $M\rightarrow M+1$ on the right-hand side. Repeating the same arguments as in the case of even $M$ is then straightforward.

\section{RBM representation of 4-Local interaction}
\label{app:4bdy_id}
In the case of four-body interaction, we can write
\begin{equation}
\begin{split}
    \exp\bigg[&-\Big(K^{(4)}\sigma_1\sigma_2\sigma_3\sigma_4+K^{(2)}_{12}\sigma_1\sigma_2+K^{(2)}_{13}\sigma_1\sigma_3\\
    &+K^{(2)}_{14}\sigma_1\sigma_4+K^{(2)}_{23}\sigma_2\sigma_3+K^{(2)}_{24}\sigma_2\sigma_4+K^{(2)}_{34}\sigma_3\sigma_4]\Big)\bigg]\\
    &=A\sum_{h=\pm 1}e^{-i(W_1\sigma_1+W_2\sigma_2+W_3\sigma_3+W_4\sigma_4)h},
    \end{split}
\end{equation}
where one possible choice of couplings is
\begin{equation}
    \begin{split}
        W\equiv W_1=W_2=W_3&=\frac{1}{2}\tan^{-1}\left[\left(1-e^{-8|K^{(4)}|}\right)^{1/4}\right],\\
W_4&=s W,\\
A&=\frac{1}{2}\left[\sec^4(2W)\sec(4W)\right]^{1/8},
    \end{split}
\end{equation}
with $s=\mathrm{sign}(K^{(4)})$ and only two-body couplings are induced
\begin{equation}
    \begin{split}
    K^{(2)}_{12}=K^{(2)}_{13}=K^{(2)}_{23}&=-\frac{1}{8}\log[\cos(4W)],\\
        K^{(2)}_{14}=K^{(2)}_{24}=K^{(2)}_{34}&=-\frac{s}{8}\log[\cos(4W)].
    \end{split}
    \end{equation}

\section{Derivation of Hadamard Actions}
\label{app:Hadamard_action}
    \begin{proof}
    The action of Hadamard operator $H^x$ at index $l$ is
    \begin{equation}
        \begin{split}
        &H^x_l \Psi_{\mathcal{L}_0}(\vec{z}) = \frac{1}{\sqrt{2}} z_l \Psi_{\mathcal{L}_0}(\vec{z})  + \frac{1}{\sqrt{2}}\Psi_{\mathcal{L}_0} (z_1,...-z_l,..)\\
        &= \frac{1}{\sqrt{2}}\sum_{\vec{h}} P_1(\vec{z},\vec{h})P_3(\vec{h}) \left(z_l + e^{-2 i a_l z_l - 2 i z_l \sum_{j} W_{lj}h_j }\right),
        \end{split}
        \label{eq:Hx_action_2}
    \end{equation}
    where we have introduced the abbreviations
\begin{equation}
\begin{split}
    P_1(\vec{z},\vec{h}) &\equiv \exp\bigg[i \left( \sum_i a_i z_i + \sum_{i,j}z_i W_{ij}h_j \right) \bigg], \\
    P_2(\vec{h})&\equiv \exp\left[ i \left( \sum_{i<j}h_iL_{ij}h_j+\sum_i b_ih_i \right) \right].
    \end{split}
\end{equation}
Our aim is to express this result in terms of a new L-DBM. We assume that the weights of the hidden units linked to visible qubit $l$ are modified as
\begin{equation}
W_{lj}\rightarrow W_{lj}+\Delta_{lj},~~~~~a_l \rightarrow a_l + \delta_l.
\end{equation}
We introduce a new hidden unit $h_{M+1}$, connected to visible qubit $l$ and laterally connected to all existing hidden units. Our L-DBM ansatz for the resulting wave function is
\begin{widetext}
\begin{equation}
    \begin{split}
    \Psi_{\mathcal{L}_1}(\vec{z}) &= \sum_{\vec{h}}P_1(\vec{z},\vec{h})P_2(\vec{h})
    \exp\left[\delta_l z_l+z_l \sum_{j=1}^M \Delta_{lj}h_j\right]\\
    &\times\sum_{h_{M+1}}\exp\left[z_l W_{l,M+1}h_{M+1}+\sum_{j=1}^M h_jL_{j,M+1} h_{M+1} + b_{M+1}h_{M+1}\right].
    \end{split}
    \label{eq:L-DBM_update}
\end{equation}
Substituting Eqs. (\ref{eq:Hx_action_2}) and (\ref{eq:L-DBM_update}) into Eq. (\ref{eq:Hx_action}) and considering the possible values $z_l=\pm 1$ yields two constraints: When $z_l=+1$ we find
\begin{equation}
\begin{split}
1+\exp\left(-2 a_l -2 \sum_{j=1}^M W_{lj} h_j\right) =A\Bigg\{&\exp\left[\left(\delta_l+ W_{l,M+1}+ b_{M+1}\right)+\sum_{j=1}^M \left(\Delta_{lj} +  L_{j,M+1}\right) h_j\right]\\
+&\exp\left[\left(\delta_l- W_{l,M+1}-b_{M+1}\right)+\sum_{j=1}^M \left(\Delta_{lj} -  L_{j,M+1}\right) h_j]\right]\Bigg\},
\end{split}
\end{equation}
and when $z_l=-1$, we find
\begin{equation}
\begin{split}
 -1+\exp\left(2 a_l + 2 \sum_{j=1}^M W_{lj} h_j\right) =A\Bigg\{&\exp\left[\left(-\delta_l - W_{l,M+1}+ b_{M+1}\right) - \sum_{j=1}^M \left(\Delta_{lj}  - L_{j,M+1}\right) h_j\right]\\
 +&\exp\left[(-\delta_l  + W_{l,M+1}- b_{M+1} ) - \sum_{j=1}^M (\Delta_{lj} + L_{j,M+1}) h_j\right]\Bigg\}.
\end{split}
\end{equation}
\end{widetext}
One possible solution is
\begin{equation}
\begin{split}
\delta_l &=  \pi/4 - a_l\\
\Delta_{lj} &= - W_{lj} 
   \end{split}
\end{equation}
\begin{equation}
\begin{split}
W_{l,M+1}&=\pi/4\\
L_{j,M+1} &= - W_{lj} \\
b_{M+1}&=-\left(a_l+\frac{\pi}{4}\right)\\
\label{eq:L-DBM_updated_couplings_appendix}
   \end{split}
\end{equation}
\begin{equation}
     A=e^{-i\pi/4}/\sqrt{2}
\end{equation}
Thus we have demonstrated by construction that one can rewrite the action of $H^x_l$ on the L-DBM wave function in terms of a different L-DBM (up to normalization).
\end{proof}

\vspace{1cm}
\section{Action of $H^y$, ${H^y}^{\dag}$ Operators on L-DBM}
\label{app:hadamard}
The action of the rotation operator $H^y_l$ on the L-DBM ansatz can be summarized as follows:
\begin{enumerate}
\item Introduce one new hidden qubit $h_{M+1}$, which couples to all existing hidden units linked to $\sigma^z_l$
\begin{equation}
    \begin{split}
        W_{l,M+1}&=i\pi/4\\
        L_{j,M+1}&=-W_{lj}\\
        b_{M+1}&=-a_l.
         \end{split}
    \end{equation}
\item Sever all existing couplings between visible qubit $l$ and the hidden layer: $W_{lj}\rightarrow 0$
\item Update the one body term at index $l$: $a_l\rightarrow -i\pi/4$
\item Multiply by the normalization factor $ C=1/\sqrt{2}$.
\end{enumerate}
The action of conjugate rotation operator ${H^y_l}^\dag$ on the L-DBM ansatz can be summarized as follows:
\begin{enumerate}
\item Introduce one new hidden qubit $h_{M+1}$, which couples to all existing hidden units linked to $\sigma^z_l$
\begin{equation}
    \begin{split}
        W_{l,M+1}&=-\left(a_l+i\pi/4\right)\\
        L_{j,M+1}&=-W_{lj}\\
        b_{M+1}&=-a_l.
         \end{split}
    \end{equation}
\item Sever all existing couplings between visible qubit $l$ and the hidden layer: $W_{lj}\rightarrow 0$
\item Update the one body term at index $l$: $a_l\rightarrow 0$
\item Multiply by the normalization factor $ C=1/\sqrt{2}$.
\end{enumerate}

\end{document}